\documentclass[12pt,a4paper]{article}

\usepackage{amsthm}
\usepackage{amsfonts}
\usepackage{amsmath}
\usepackage[top=20mm, bottom=20mm, left=21mm, right=20mm]{geometry}

\def\R{\mathbb{R}}
\def\C{\mathbb{C}}
\def\col{\mathrm{col}\,}

\numberwithin{equation}{section}
\theoremstyle{plain}
\newtheorem{theorem}{Theorem}[section]

\newtheorem{remark}{Remark}[section]

\newtheorem*{corollary*}{Corollary}

\def\R{\mathbb{R}}
\def\col{\mathrm{col}\,}

\date{}
\begin{document}
	
	%\title{Самосопряжённые граничные задачи для систем Дирака с нелокальными потенциалами.}

\title{Well-posed Cauchy problem\\ and the Hamiltonian form of (2+1) nonlinear equations\\  integrable  by inverse scattering transform}
	
	\author{L.~P.~Nizhnik\footnote{Institute of Mathematics NAS of Ukraine, Kyiv, Ukraine, nizhnik@imath.kiev.ua }}
	
	\maketitle

	\begin{abstract}
%\large
	The Hamiltonian form of the (2+1) nonlinear integrable Schrödinger equation and the system of two (2+1) nonlinear analogue of the mKdV equation is proved.
	A well--posed Cauchy problem is formulated and the solvability of such a problem for the (2+1) nonlinear analogue of the mKdV equation is proved.
	\end{abstract}
%\section*
%\large

\section{The nonlinear Schr\"odinger equation}

The paper \cite{1} presents  well-posedness and solvability of the Cauchy problem for the (2+1) nonlinear Schrödinger equation of the form:
 \begin{equation}\label{1.1}
 	i\frac{\partial u}{\partial t}+\frac{\partial^2 u}{\partial x^2}+\frac{\partial^2 u}{\partial y^2}+(v_1 +v_2)u=0,
 \end{equation}
where pseudopotentials $v_1$ and $v_2$ are real-valued functions related to the solution $u(x,y;t)$ of equation (\ref{1.1}) in the form:
 \begin{equation}\label{1.2}
 	\frac{\partial v_1}{\partial x}=2\frac{\partial}{\partial y}|u|^2,\qquad \frac{\partial v_2}{\partial y}=2\frac{ \partial}{\partial x}|u|^2.
 \end{equation}
 For the Cauchy problem, equality (\ref{1.2}) can be represented as following:
  \begin{equation}\label{1.3}
  \begin{array}{l}
\displaystyle  v_1(x,y;t)=p_-(y,t)+2 \int\limits_{-\infty}^{x}\,\frac{\partial}{\partial y}|u(s,y;t |^2\,ds,\\[3mm]
\displaystyle v_2(x,y;t)=q_+(x,t)-2\int\limits_{y}^{+\infty}\,\frac{\partial}{\partial x }|u(x,s;t|^2\,ds,
 \end{array}
 \end{equation}
 where
% \begin{equation}\label{}
$$ p_{-}(y,t)=v_1(- \infty,y;t),\qquad q_{+}(x,t)=v_2(x,+\infty;t).$$
% \end{equation}
The Cauchy problem for equations (\ref{1.1})--(\ref{1.3}) consists of finding a solution $u(x,y;t)$ of equation (\ref{1.1})--(\ref{1.3}) with a given initial data $u(x,y;0)$
and specifying two functions $ p_{-}(y,t)$ and $q_{+}(x,t)$. A solution to such a Cauchy problem is given in \cite{1}.
% IST на основе обратной задачи рассеяния для пространственно--двумерной системы Дирака $L\psi=0$, где
%\begin{equation}\label{Nizheq:opD}
%$$L=	\begin{pmatrix}\displaystyle \frac{\partial}{\partial \xi}& \displaystyle u\\\displaystyle -\bar{u} & \displaystyle \frac{\partial}{\partial \eta}\end{pmatrix},$$
%\end{equation}
%переводит функцию $u(x,y;t)$ в функцию $f(x,y;t)$:\,   $f=\mathfrak{A}u$, где функции $u,f\in L_2(|R^2)$. Явный вид операторов $\mathfrak{A}$ и $\mathfrak{A}^{-1}$ указан в работах %\cite{1,3}. Это даёт существование и единственность решения задачи Коши, поскольку данные рассеяния $f$ удовлетворяют линейное уравнение
%\begin{equation}\label{1.4}
%	i \frac{\partial \overline{ f}}{\partial t}+\frac{\partial^2 \overline{f}}{\partial \xi^2}+\frac{\partial^2 \overline{ f}}{\partial \eta^2}+(p_+(\eta,t)+q_-(\xi,t)) \overline{f}=0.
%\end{equation}

In order to construct the Hamiltonian form for the equation (\ref{1.1})--(\ref{1.2}),
we proceed in the same way as in the case of the (1+1) Schrödinger equations. We will denote by $u_1=u$, and by $u_2=\bar{u}$ in equalities (\ref{1.1})--(\ref{1.2}).
The relationship (\ref{1.2}) of the pseudopotentials $v_1$ and $v_2$ with $u_1$ and $u_2$ can be represented in the form:
\begin{equation}\label{1.5}
v_1(x,y;t)=r(y,t)+2\partial_y\partial_x^{-1}(u_1u_2),\quad	v_2(x,y;t)=s(x,t)+2\partial_x\partial_y^{-1}(u_1u_2),
\end{equation}
where functions $r$ and $s$ are real-valued and uniformly bounded in their arguments:
\begin{equation}\label{1.6}
\begin{array}{l}
\displaystyle r(y,t)=\frac{1}{2}[v_1(+\infty,y;t)+v_1(- \infty,y;t)],\\[3mm]
\displaystyle s(x,t)=\frac{1}{2}[v_2(x,+\infty;t)+v_2(x,-\infty;t)].
%\end{equation}	
%\begin{equation}\label{1.7}
\end{array}
\end{equation}
In formulae (\ref{1.5}), the operators $\partial_x^{-1}$ and $\partial_y^{-1}$ are skew-symmetric, having the following form:
\begin{equation}\label{1.8}
\begin{array}{l}
\displaystyle \partial_x^{-1}	f(x)=\frac{1}{2}\Bigl[\int\limits_{-\infty}^x\,f(s)\,ds-\int\limits_{x}^{+\infty} \,f(s)\,ds\Bigr],\\ [3mm]
\displaystyle \partial_y^{-1}g(y)=\frac{1}{2}\Bigl[\int\limits_{-\infty}^y\,g(s)\,ds-\int\limits_{y}^{+\infty} \,g(s)\,ds\Bigr],\quad f,g \in L_1(\R^1).
\end{array}	
\end{equation}
Therefore, the identity holds:
\begin{equation}\label{1.9}
\iint\limits_{\R^2}\,\psi(x,y)[\partial_x \partial_y^{-1}]\varphi(x,y)\,dx dy=\iint\limits_{\R^2}\,\varphi(x,y)[\partial_x \partial_y^{-1}]\psi(x,y)\,dx dy.
\end{equation}
A similar identity is also valid for the operator $\partial_y \partial_x^{-1}$. The functions $\psi$ and $\varphi$ are assumed to be rapidly decreasing at infinity  with their derivatives $\psi,\,\varphi \in W_2'(\R^2).$
\begin{theorem}
Equation (\ref{1.1}) with explicit relation (\ref{1.5}) of values $v_1$ and $v_2$ with solution $u_1,$\,$u_2$ admits a representation of the form:
\begin{equation}\label{1.10}
\frac{\partial u_1}{\partial t}=\frac{-i\delta H}{\delta u_2},\qquad \frac{\partial u_2}{\partial t}=\frac{i\delta H}{\delta u_1},
\end{equation}
where the Hamilton function $H$ has the form:
\begin{equation}\label{1.11}
	H=\iint\limits_{\R^2}\,[ u_{1,x}u_{2,x}+u_{1,y}u_{2,y}-u_1(r+s)u_2-u_1u_2[\partial_x \partial_y^{-1}+\partial_y \partial_x^{-1}] (u_1u_2)]\,dx\, dy.
\end{equation}
\end{theorem}
\begin{proof}
Using the property (\ref{1.9}), we obtain:
$$\frac{\delta H}{\delta u_2}=[-\Delta u_1-(r+s)u_1 -2u_1[\partial_x \partial_y^{-1}+\partial_y \partial_x^{-1}]  (u_1u_2)].
$$
Considering the explicit relation (\ref{1.5}) of $v_1$ and $v_2$ with the solution $u_1,$\,$u_2$, we have
$$\frac{\delta H}{\delta u_2}= -[\Delta u_1+(v_1+v_2)u_1].
$$
Therefore, the equation $\displaystyle \frac{\partial u_1}{\partial t}=\frac{-i\delta H}{\delta u_2}$ is equivalent to (\ref{1.1}) because of $u_1=u$, and $u_2=\bar{u}$.
Similarly, it is proved that the equation $\displaystyle \frac{\partial u_2}{\partial t}=\frac{i\delta H}{\delta u_1}$ is equivalent to the equation (\ref{1.1}) for $u_2=\bar{u}$.
\end{proof}
\section{Spatially two-dimensional nonlinear integrable analogues of the KdV equation}
Back in 1980, in the work \cite{2}, a spatially symmetric two--dimensional  KdV equation was considered, which in the literature became known as the Nizhnik-Novikov-Veselov equation \cite{4}.
This equation has the form:
\begin{equation}\label{2.1}
\begin{array}{l}
 \displaystyle   \frac{\partial u}{\partial t}=k_1\frac{\partial^3 u}{\partial x^3}+ k_2\frac{\partial^3 u}{\partial y^3}+3\frac{\partial}{\partial x}(v_1u)+3\frac{\partial}{\partial y}(v_2u),\\[5mm]
 \displaystyle \frac{\partial v_1}{\partial y}=k_1\frac{\partial u}{\partial x},\qquad  \frac{\partial v_2}{\partial x}=k_2\frac{\partial u}{\partial y}.
\end{array}
\end{equation}
It has a Lax representation:
\begin{equation}\label{2.2}
  LP - QL=0
\end{equation}
for the operator:
\begin{equation}\label{2.3}
  L=\frac{\partial^2}{\partial x \partial y}+u(x,y;t).
\end{equation}
Let us now consider the spatially two--dimensional analogue of the modified KdV equation \cite{3}, when the nonlinearity is cubic, in contrast to equation (\ref{2.1}), where it is quadratic. The equation has the form:
\begin{equation}\label{2.4}
  \begin{array}{l}
  \displaystyle   \frac{\partial u}{\partial t}=\frac{\partial^3 u}{\partial x^3}+\frac{\partial^3 u}{\partial y^3}+(vu)_y+(wu)_x-\frac{1}{2}(v_y+w_x)u, \\[3mm]
    v_x=3(u^2)_y,\qquad w_y=3(u^2)_x.
  \end{array}
\end{equation}
and the Lax representation has the form (\ref{2.2}), where
\begin{equation}\label{2.5}
\begin{array}{l}
  L= \begin{pmatrix} \displaystyle \frac{\partial}{\partial x}& u\\-u &  \displaystyle \frac{\partial}{\partial y}\end{pmatrix},\,\, P= \begin{pmatrix} \displaystyle \mathfrak{D}-\frac{1}{2}v_y& -3u_x\partial_x\\ \displaystyle 3u_y\partial_y &  \displaystyle \mathfrak{D}- \frac{1}{2}w_x\end{pmatrix},\\[10mm]
  Q= \begin{pmatrix}\displaystyle \mathfrak{D}-\frac{1}{2}v_y-w_x & 3\partial_y u_y\\ \displaystyle -3\partial_x u_x &  \displaystyle \mathfrak{D}-\frac{1}{2}w_x-v_y\end{pmatrix},
  \end{array}
\end{equation}
where
$$ \mathfrak{D}=\partial_t-\partial^3_x-\partial^3_y-v\partial_y-w\partial_x.
$$
\begin{theorem}\label{th:2.1}
Let $u(x,y;t)\in L_2(\R^2)$ be a solution of equation (\ref{2.4}), and let the functions $v$ and $w$ of equation (\ref{2.4}) be related to the solution $u$ of equation (\ref{2.4}) as following:
\begin{equation}\label{2.6d}
\begin{array}{l}
\displaystyle v(x,y;t)=p(y,t)+3\int\limits_{-\infty}^{x}\,[u^2(s,y,t)]_y'\,ds,\\ [3mm]
\displaystyle w(x,y;t)=q(x,t)- 3\int\limits_{y}^{+\infty}\,[u^2(x,s,t)]_x'\,ds,
\end{array}
\end{equation}
where
$$ p(y,t)=v(-\infty,y,t),\qquad q(x,t)=w(x,+\infty,t).
$$
Then, the scattering data $f(x,y,t)$ satisfy the equation:
\begin{equation}\label{2.10}
  \frac{\partial f}{\partial t}-\frac{\partial^3 f}{\partial x^3}-\frac{\partial^3 f}{\partial y^3}-q(x,y)\frac{\partial f}{\partial x}-\frac{1}{2}q_x'(x,t)f-p(y,t)\frac{\partial f}{\partial y}-\frac{1}{2}p_y'(y,t)f=0.
\end{equation}
\end{theorem}
\begin{proof}
Recall \cite{3} that in the case $u_1,\,u_2 \in L_2(\R^2)$ the scattering problem for the equation
\begin{equation}\label{eq:sys}
 \begin{pmatrix}\frac{\partial}{\partial x}& u_1\\u_2 & \frac{\partial}{\partial y}\end{pmatrix} \begin{pmatrix}\psi_1\\ \psi_2 \end{pmatrix}=0
\end{equation}
consists in constructing solutions $\psi_1$,\,$\psi_2$ in the form:
\begin{equation}\label{2.7}
\begin{array}{l}
  \psi_1(x,y)=a_1(y)+o(1),\,\,\textrm{where}\,\, x\rightarrow-\infty;\,\,\, \psi_1(x,y)=b_1(y)+o(1),\,\,\textrm{where}\,\,x\rightarrow+\infty; \\[3mm]
   \psi_2(x,y)=a_2(x)+o(1),\,\,\textrm{where}\,\, y\rightarrow-\infty;\,\,\,  \psi_2(x,y)=b_2(x)+o(1),\,\,\textrm{where}\,\, y\rightarrow+\infty.
\end{array}
\end{equation}
The functions $a_1$,\,$a_2 \in L_2(\R^1)$ are incident waves, and the functions $b_1$,\,$b_2 \in L_2(\R^1)$ are scattered waves. In this case, if
$a=\col(a_1,a_2)$ are given, then $b=\col(b_1,b_2)$ are uniquely determined by $b=Sa$, where $S=I+F$ is the scattering operator, and
$$F= \begin{pmatrix}F_{11}& F_{12}\\F_{21} & F_{22}\end{pmatrix}
$$ is the matrix of Hilbert--Schmidt integral operators $F_{ij}.$ There exists an operator
$S^{-1}=I+G$.
The scattering data are a pair of operators $(F_{21},\,G_{12})$ or $(F_{12},\,G_{21})$. The scattering data are uniquely determined by the potentials $u_1,\,u_2$. The potentials are uniquely reconstructed from the scattering data. The operator that transforms the potentials into scattering data is denoted by $\mathfrak{A}$, and the solutions to the inverse scattering problem are determined by the operator
 $\mathfrak{A}^{-1}$.
 In the case $u_2=-\bar{u}_1$ for the scattering data, only one operator $F_{21}$ is sufficient to solve the inverse scattering problem. In this case, the necessary and sufficient condition for the Hilbert-- Schmidt operator $F_{21}$ to be scattering data is the condition $||F_{21}||< 1$.
 In what follows, we will omit the indices (2,1) of the operator $F_{21}$, and denote the kernel of this operator by $f(x,y)$. Then, if $a_2=0$, then $b_2=Fa_1$.
 Let $u(x,y;t)$ be a solution to system (\ref{2.4}). Let $\psi=\col(\psi_1,\psi_2)$ be a solution to equation (\ref{eq:sys}).
 Then the function $\hat{\psi}=P\psi$ by virtue of the Lax representation (\ref{2.2}) also satisfies the equation $L\hat{\psi}=0$. Let the incident waves for the solution $\psi$ have the form $(a_1, 0)$. Then, the incident waves $(\hat{a}_1,\hat{a}_2)$ and the scattered one $\hat{b}_2$ for the solution $\hat{\psi}$ have the form:
\begin{equation}\label{2.8}
\begin{array}{l}
\displaystyle  \hat{a}_1(y,t)=[-\frac{\partial^3}{\partial y^3}-\frac{1}{2}p_y'(y,t)-p(y,t)\frac{\partial}{\partial y}]a_1(y),\quad \hat{a}_2=0, \\[5mm]
 \displaystyle \hat{b}_2(x,t)=[\frac{\partial}{\partial t}-\partial_x^3-q(x,t)\frac{\partial}{\partial x}-\frac{1}{2}q_x'(x,t)]b_2(x,t).
\end{array}
\end{equation}
Since $b_2=Fa_1$ and $\hat{b}_2=F\hat{a}_2$, then  equation (\ref{2.10})  follows from (\ref{2.8}) for the scattering data $f(x,y;t)$
of the kernel of the integral operator $F(t)$.
\end{proof}
\begin{theorem}\label{th:2.2}
Let the initial conditions of equation (\ref{2.10}) be $f(x,y,0) \in L_2(\R^2)$, and let the real functions $p(y,t)$ and $q(x,t)$ belong to the space $\C^1(\R^2)$. Then the solution $f(x,y,t)$ of the Cauchy problem for equation (\ref{2.10}) exists and is unique for any time interval, and the integral operator $F(t)$ with kernel $f(x,y,t)$ has the property
\begin{equation}\label{2.11}
  ||F(t)||_{L_2}=||F(0)||_{L_2}.
\end{equation}
\end{theorem}
\begin{proof}
The evolution equation (\ref{2.10}) can be represented as
$$\frac{\partial f}{\partial t}=[A_1(t)+A_2(t)]f,
$$
where  the operators
$\displaystyle A_1=\frac{\partial^3}{\partial x^3}+q\frac{\partial}{\partial x}+\frac{1}{2}q_x'$ and $\displaystyle A_2=\frac{\partial^3}{\partial y^3}+p\frac{\partial}{\partial y}+\frac{1}{2}p_y'$ are skew-symmetric
$A_j^*=-A_j$.  The operator $A_1(t)$ acts on the variable $x$, and the operator $A_2(t)$ acts on the variable $y$, and these operators commute and smoothly depend on $t$. Therefore, there exist unitary operators $U_1(t)$ and $U_2(t)$, as solutions of the equations $\displaystyle \frac{d U_j(t)}{dt}=A_j(t)U_j(t);$\,\,$U_j(0)=I.$
%действия которых такие, что

The solution of the Cauchy problem for equation (\ref{2.10}) can be represented in the form:
$$f(x,y,t)=U_1(t)U_2(t)f(x,y,0).
$$
This shows that the equality (\ref{2.11}) is true.
\end{proof}
\begin{theorem}\label{th:2.3}
There is a unique classical solution to the Cauchy problem for the system (\ref{2.4})
with given sufficiently smooth initial condition $u(x,y,0) \in L_2(\R^2)$ and given sufficiently smooth functions $p(y,t),$\,
 $q(x,t)\in \C^1(\R^2)$.
\end{theorem}
\begin{proof}
Since the solution of the Cauchy problem is reduced by the IST method to the linear Cauchy problem for equation (\ref{2.10}), the initial data $f(x,y,0)=\mathfrak{A}u(x,y,0)$ obtained by the IST method.
The inequality $||F(0)||_{L_2}<1$ is valid, and therefore, we have $||F(t)||_{L_2}<1$ by virtue of theorem \ref{th:2.2}. Hence, the function
$u(x,y,t)=\mathfrak{A}^{-1}f(x,y,t)$ is a solution of the Cauchy problem  by virtue of theorem \ref{th:2.2}.
\end{proof}
\begin{remark}
Using the IST method, one can remove the requirement for additional smoothness of the initial data of the Cauchy problem  in Theorem \ref{th:2.3} if one defines a generalized solution to the Cauchy problem, so that in the case where the formula $u(x,y,t)=\mathfrak{A}^{-1}f(x,y,t)$ has meaning, the function $u(x,y,t)$ is a generalized solution to the Cauchy problem.
%пЭтот подход реализован\cite{Tat}.
\end{remark}

\section{Hamiltonian form of a system of two (2+1) nonlinear analogues of mKdV}
Let us now consider a system of two equations with respect to the functions $u_1(x,y,t)$ and $u_2(x,y,t)$
\begin{equation}\label{3.1}
\begin{array}{l}
\displaystyle  \frac{\partial u_1}{\partial t}=\Bigl(\frac{\partial^3}{\partial x^3}+\frac{\partial^3}{\partial y^3}\Bigr)u_1-v \frac{\partial u_1}{\partial y}-w \frac{\partial u_1}{\partial x}-(v_1+w_1)u_1, \\[7mm]
\displaystyle   \frac{\partial u_2}{\partial t}=\Bigl(\frac{\partial^3}{\partial x^3}+\frac{\partial^3}{\partial y^3}\Bigr)u_2-v \frac{\partial u_2}{\partial y}-w \frac{\partial u_2}{\partial x}-(v_2+w_2)u_2.
\end{array}
\end{equation}

Pseudopotentials $v_1,\,v_2,\,v$,\,$w_1,\,w_2,\,w$ are related to the solutions $u_1$,\,$u_2$ by equalities:
\begin{equation}\label{3.2}
\begin{array}{l}
\displaystyle   \frac{\partial v_1}{\partial x}=3(u_{1,y}'\cdot u_2)_y',\quad  \frac{\partial w_1}{\partial y}=3(u_{1,x}'\cdot u_2)_x',\\[5mm]
 \displaystyle   \frac{\partial v_2}{\partial x}=3(u_{1}\cdot u_{2,y}')_y',\quad  \frac{\partial w_2}{\partial y}=3(u_{1}\cdot u_{2,x}')_x',\\[5mm]
 \displaystyle   \frac{\partial v}{\partial x}=3(u_{1}\cdot u_{2})_y',\qquad  \frac{\partial w}{\partial y}=3(u_{1}\cdot u_{2})_x'.
\end{array}
\end{equation}

System (\ref{3.1})--(\ref{3.2}) admits the Lax representation (\ref{2.2}).
In this case, the operator $$L= \begin{pmatrix}\frac{\partial}{\partial x}& u_1\\ u_2 & \frac{\partial}{\partial y}\end{pmatrix},$$ and the operators $P$ and $Q$ are matrix operators of the form:
\begin{equation}\label{3.3}
  P= \begin{pmatrix} \mathfrak{D}+v_1& -3u_{1,x}\partial_x\\ -3u_{2,y}\partial_y & \mathfrak{D}+w_2\end{pmatrix},\quad
  Q= \begin{pmatrix} \mathfrak{D}+w_1+w_2+v_1 & 3 u_{1,y}\partial_y+3u_{1,yy}\\ 3 u_{2,x}\partial_x +3u_{2,xx}& \mathfrak{D}+w_2+v_1+v_2\end{pmatrix},
  \end{equation}
where
$$ \mathfrak{D}=\partial_t-\partial^3_{xxx}-\partial^3_{yyy}+w\partial_x+v\partial_y.
$$
\begin{theorem}
Let the pseudopotentials of (\ref{3.2}) be uniquely related to solutions $u_1$ and $u_2$ by equalities:
\begin{equation}\label{3.4}
\begin{array}{l}
 \displaystyle  v_1=3\partial^{-1}_{x}(u_{1,y}\cdot u_2)_y', \quad w_1=3\partial^{-1}_{y}(u_{1,x}\cdot u_2)_x', \\[3mm]
  \displaystyle  v_2=3\partial^{-1}_{x}(u_{1}\cdot u_{2,y})_y', \quad w_2=3\partial^{-1}_{y}(u_{1}\cdot u_{2,x})_x', \\[3mm]
  \displaystyle  v=3\partial^{-1}_{x}(u_{1}\cdot u_{2})_y', \qquad w=3\partial^{-1}_{y}(u_{1}\cdot u_{2})_x',
\end{array}
\end{equation}
where $\partial^{-1}_{x}$ and $\partial^{-1}_{y}$ are skew--symmetric operators defined by equalities
(\ref{1.8}). Then, the system (\ref{3.1})--(\ref{3.4}) is a Hamiltonian system of the form:
\begin{equation}\label{3.5}
\frac{\partial u_1}{\partial t}=\frac{-\delta H}{\delta u_2},\qquad \frac{\partial u_2}{\partial t}=\frac{\delta H}{\delta u_1},
\end{equation}
where Hamilton function $H=H_1+H_2+H_3,$
and
\begin{equation}\label{3.6}
\begin{array}{l}
\displaystyle H_1= \iint\limits_{\R^2}\,[u_{1,xx}\cdot u_{2,x}+u_{1,yy}\cdot u_{2,y}]\,dx\,dy,\\[3mm]
\displaystyle H_2=  3 \iint\limits_{\R^2}\,[u_{1,y}\cdot u_{2}\partial^{-1}_{x}(u_{1}\cdot u_{2,y})]\,dx\,dy,\\[3mm]
\displaystyle  H_3= 3 \iint\limits_{\R^2}\,[u_{1,x}\cdot u_{2}\partial^{-1}_{y}(u_{1}\cdot u_{2,x})]\,dx\,dy.
\end{array}
\end{equation}
\end{theorem}
\begin{proof}
Functional $H_1$ can be represented in the following two forms:
%\begin{equation}\label{}
  $$
  H_1=\iint\limits_{\R^2}\,[u_1(u_{2,xxx}+u_{2,yyy})]\,dx\,dy=-\iint\limits_{\R^2}\,[(u_{1,xxx}+u_{1,yyy})u_2]\,dx\,dy.
%\end{equation}
$$
It follows from the first representation: $\displaystyle \frac{\delta H_1}{\delta u_1}=\frac{\partial^3u_2}{\partial x^3}+\frac{\partial^3 u_2}{\partial y^3},$
and from the second one:
$$\frac{\delta H_1}{\delta u_2}=-\Bigl[\frac{\partial^3u_1}{\partial x^3}+\frac{\partial^3 u_1}{\partial y^3}\Bigr].
$$
Finding $\displaystyle \frac{\delta H_2}{\delta u_2},$ it is desirable to represent functional $H_2$ in the form:
%\begin
$$\hat{H_2}=3\iint\limits_{\R^2}\,[(u_{1,y}\cdot u_{2})\partial^{-1}_{x}(u_{1}\cdot u_{3,y})]\,dx\,dy\Big|_{u_3=u_2}=3\iint\limits_{\R^2}\,u_3\partial_y[u_1\partial^{-1}_{x}(u_{1,y}\cdot u_{2})]\,dx\,dy\Big|_{u_3=u_2}.
$$
Then,
%\begin
$$
\begin{array}{l}
\displaystyle \frac{\delta H_2}{\delta u_2}=\displaystyle \Bigl[ \frac{\delta \hat{H}_2}{\delta u_2}+ \frac{\delta \hat{H}_2}{\delta u_3}\Bigr]\Big|_{u_3=u_2}=3u_{1,y}\partial^{-1}_{x}(u_{1}\cdot u_{2,y})+3 \partial_y[u_1\partial^{-1}_{x}(u_{1,y}\cdot u_{2})]=\\[5mm]
=3u_{1,y}\partial^{-1}_{x}(u_{1}\cdot u_2)_y'+3u_1\partial_y \partial^{-1}_{x}(u_{1,y}\cdot u_{2}),
\end{array}
$$
and by virtue of (\ref{3.4}):
$$\displaystyle \frac{\delta H_2}{\delta u_2}=v\cdot u_{1,y}+v_1\cdot u_{1}.$$
Similarly we get:
%\begin
$$\frac{\delta H_3}{\delta u_2}=w\cdot u_{1,x}+w_2\cdot u_1.$$
Therefore,
$$\frac{\delta H}{\delta u_2}=-\Bigl(\frac{\partial^3u_1}{\partial x^3}+\frac{\partial^3 u_1}{\partial y^3}\Bigr)+w\cdot u_{1,x}+v\cdot u_{1,y}+(v_1+w_2)\cdot u_{1}.
$$
The first equation of (\ref{3.1}) reduces to the equation:
$\displaystyle \frac{\partial u_1}{\partial t}=-\frac{\delta H}{\delta u_2}.
$
It is similarly verified that the second equation of (\ref{3.1}) is equivalent to the equation
$\displaystyle \frac{\partial u_2}{\partial t}=\frac{\delta H}{\delta u_1}.
$
\end{proof}
{\bf Acknowledgments.}\\
 The author is  grateful  to the Simons Foundation  for the financial support of the Institute of Mathematics of the National Academy of Sciences of Ukraine
  (1030291 and 1290607, L.N.), which facilitated the preparation of this English version  of the results, partially published only in Russian before 1992.

  % которая способствовала подготовке этой английской версии изложение результатов, частично опубликованных только на русском языке ещё до 1992 года.

\end{document}